\documentclass[11pt]{article}
\usepackage[margin=1.05in]{geometry}
\usepackage{amsmath,amssymb,amsthm,mathtools}
\usepackage{microtype}
\usepackage[hidelinks]{hyperref}
\newtheorem{theorem}{Theorem}[section]
\newtheorem{proposition}[theorem]{Proposition}
\newtheorem{lemma}[theorem]{Lemma}
\newtheorem{corollary}[theorem]{Corollary}
\theoremstyle{remark}

\newcommand{\E}{\mathbb E}
\newcommand{\R}{\mathbb R}

\newcommand{\Cov}{\operatorname{Cov}}
\newcommand{\Var}{\operatorname{Var}}
\newcommand{\arctanh}{\operatorname{arctanh}}
\title{Critical and near-critical influence bounds\\for ferromagnetic Ising models}
\author{Yan Ru Pei}
\date{September 14, 2026}
\hypersetup{pdftitle={Critical and near-critical influence bounds for ferromagnetic Ising models},pdfauthor={Yan Ru Pei}}
\begin{document}
\maketitle
\begin{abstract}
For a ferromagnetic Ising model on a graph of maximum degree $\Delta\ge3$,
we prove a bound of order $\sqrt n$ on every row of the influence matrix
at the tree uniqueness threshold. The estimate is uniform in the degree,
the external fields, and all pinnings. More generally, if the couplings
are bounded by $\beta$ and
$\varepsilon=((\Delta-1)\tanh\beta-1)_+$, the bound is
$C(\sqrt n+n\varepsilon)$.
The proof combines a pointwise cavity bound with a positive-series
magnetization tilt and the field comparison theorem of Ding, Song and Sun.
The critical estimate removes the logarithm in recent general graphical
bounds for the ferromagnetic case.
As a consequence, zero-field single-site Glauber dynamics mixes in
polynomial time throughout the supercritical window
$\varepsilon=O(\sqrt{\log n/n})$, with the polynomial degree depending on the
window size.
\end{abstract}

\section{Introduction and statements}

At the uniqueness threshold of the regular tree, Ising correlations need
not have a uniformly summable tree bound. A finite graph nevertheless
constrains how large their total influence can be. The expected critical
scale is $\sqrt n$, and this scale governs spectral independence and
its applications to Glauber dynamics.
Chen, Chen, Yin and Zhang~\cite{CCYZ} studied this endpoint and retained
the square-root graphical bound as a conjecture after retracting a
coupling argument in their third version.
Recent work of Bencs and coauthors~\cite{NBW} gives an
$O_\Delta(\sqrt n\log n)$ bound on regular graphs through
non-backtracking walks, and the sharp square-root order when a
fixed spectral gap is available.
Here we obtain the exact square-root order for ferromagnets, together
with a finite-size estimate above the threshold.

Let $G=(V,E)$ be a finite simple graph, $|V|=n\ge1$, with maximum
degree at most an integer $\Delta\ge3$.
For edge couplings $J_e\ge0$ and fields $a\in\R^V$, write
\[
 \mu_a(\sigma)=\frac1{Z_a}
 \exp\left\{\sum_{uv\in E}J_{uv}\sigma_u\sigma_v
                     +\sum_{v\in V}a_v\sigma_v\right\},
 \qquad \sigma\in\{-1,1\}^V .
\]
We use the source-first influence convention
\begin{equation}\label{eq:influence}
 \Psi_a(u,v)=\mu_a(\sigma_v=1\mid\sigma_u=1)
                   -\mu_a(\sigma_v=1\mid\sigma_u=-1),
 \qquad \Psi_a(u,u)=1.
\end{equation}
All entries are nonnegative. Thus $\|\Psi_a\|_\infty$ is the maximum
row sum. If $D_a=\operatorname{diag}(\Var_{\mu_a}\sigma_u)$, then
$\Psi_a=D_a^{-1}\Cov_{\mu_a}(\sigma)$ is similar to the symmetric matrix
$D_a^{-1/2}\Cov_{\mu_a}(\sigma)D_a^{-1/2}$.
Our row bounds in particular bound its largest eigenvalue.
They therefore give spectral independence, with the unit diagonal
included in the convention.

\begin{theorem}\label{thm:critical}
Suppose $0\le J_e\le\beta_c$, where
$\beta_c=\arctanh(1/(\Delta-1))$.
For every finite field vector and every pinning leaving $k\ge1$ vertices,
the influence matrix of the conditional measure satisfies
\[
 \|\Psi\|_\infty\le4\sqrt{6e}\,\sqrt k<17\sqrt k .
\]
\end{theorem}

The exponent $1/2$ is sharp: the critical examples in
\cite[Corollary~5.3]{CCYZ} have largest influence eigenvalue of this
order. The constant here is independent of $\Delta$.

\begin{theorem}\label{thm:crossover}
Suppose $0\le\beta<\infty$, $0\le J_e\le\beta$, and put
$\varepsilon=((\Delta-1)\tanh\beta-1)_+$.
For every finite field vector and every pinning leaving $k\ge1$ vertices,
\begin{equation}\label{eq:crossover}
 \|\Psi\|_\infty
       \le \min\{k,\,512e(\sqrt k+k\varepsilon)\}.
\end{equation}
In particular, if
$(\Delta-1)\tanh\beta\le1+A/\sqrt n$ for some fixed $A\ge0$,
the bound is $512e(1+A)\sqrt k$ after every pinning.
\end{theorem}

Here is the corresponding dynamical consequence.
Let $T_{\rm mix}$ denote worst-state total-variation mixing time at
distance $1/4$ for discrete-time, random-scan heat-bath dynamics.
\begin{corollary}\label{cor:mixing}
In zero external field, under the hypotheses of
Theorem~\ref{thm:crossover}, suppose $\varepsilon\le1/64$ and put
$A=\sqrt n\,\varepsilon$ and $C=512e$.
There is a numerical constant $K$ such that
\[
 T_{\rm mix}\le
 K n^{\,2+2/(\Delta-2)}\log(en)
            \exp\{5C(A+A^2/2)\}.
\]
Consequently, for each fixed $B\ge0$, the window
$\varepsilon\le B\sqrt{\log n/n}$ has polynomial mixing time for all
sufficiently large $n$ depending on $B$.
More explicitly, the bound is
$O(n^{\,2+2/(\Delta-2)+5CB^2}\log(en))$, with a constant depending on $B$.
\end{corollary}

Polynomial mixing in a sparse critical window is an explicit question
in the introduction of Galanis, \v{S}tefankovi\v{c} and Vigoda~\cite{GSV}.
Corollary~\ref{cor:mixing} gives the ferromagnetic Ising case.
The proof does not yield the optimal mixing exponent.

The proof uses a small uniform positive field as a probe of zero-field
susceptibility. A cavity inequality from the magnetization argument of
Carlson, Davies, Kolla and Perkins~\cite{CDKP} controls the magnetization
at each vertex by $Ch^{1/3}$ at criticality.
Positivity of all terms in a spin expansion then gives
$h\chi_v\le m_v(h)\E_0e^{hM}$, where $M$ is total magnetization.
Choosing $h$ of order $n^{-3/4}$ gives the square-root bound on
each susceptibility row. Ding, Song and Sun's comparison
theorem~\cite{DSS} transfers it to arbitrary fields.
Above criticality the same argument uses
$m_v(h)\le C(\sqrt\varepsilon+h^{1/3})$ and a smaller probe field.

\section{A pointwise cavity estimate}

We first work in a uniform field $h>0$. Denote the magnetization at $u$
in a graph $H$ by $m_u^H$, and suppress $h$ in the notation.
The Griffiths correlation inequalities imply that spin-product
expectations increase when nonnegative edges are added.
These are the classical inequalities used in~\cite[Section~4]{CDKP}.

\begin{lemma}\label{lem:cavity}
Suppose every edge has coupling $\beta$, and let
$t=\tanh\beta$ and $f_\beta(x)=\arctanh(t\tanh x)$.
For each oriented edge put
$L_{u\to v}=\arctanh(m_u^{G-uv})$. Then
\begin{align}
 L_{u\to v}
    &\le h+\sum_{w\in N(u)\setminus\{v\}}f_\beta(L_{w\to u}),
       \label{eq:cavity}\\
 \arctanh(m_u^G)
    &\le h+\sum_{w\in N(u)}f_\beta(L_{w\to u}).
       \label{eq:root}
\end{align}
\end{lemma}
\begin{proof}
For an edge $uv$ of $H$, let
$A=m_u^{H-uv}$, $B=m_v^{H-uv}$, and
$C=\E_{H-uv,h}[\sigma_u\sigma_v]$.
Adding the edge gives the exact identity
$m_u^H=(A+tB)/(1+tC)$.
Since $C\ge AB$ and $A+tB\ge0$,
\[
 \arctanh(m_u^H)\le\arctanh A+\arctanh(tB).
\]
Remove the edges incident to $u$ one at a time.
At the removal of $uw$, monotonicity bounds the neighbor
magnetization by $m_w^{G-uw}$.
The final isolated vertex has magnetization $\tanh h$.
Starting with $G-uv$ or $G$ proves the two inequalities.
\end{proof}

\begin{proposition}\label{prop:response}
Uniformly over the graph, its size, and $\Delta\ge3$,
\begin{enumerate}
\item if $J_e\le\beta_c$, then $m_u(h)\le4h^{1/3}$ for every $u$ and $h>0$;
\item if $J_e\le\beta$ and $0\le\varepsilon=(\Delta-1)\tanh\beta-1
 \le1/64$, then
$m_u(h)\le16(\sqrt\varepsilon+h^{1/3})$ for every $u$ and $h>0$.
\end{enumerate}
\end{proposition}
\begin{proof}
By coupling monotonicity it suffices to use uniform coupling at its
upper bound. Isolated vertices satisfy the assertions directly.
Otherwise put $d=\Delta-1$ and
$X=\max_{u\to v}L_{u\to v}$, a finite attained maximum.
At $\beta=\beta_c$, Lemma~\ref{lem:cavity} gives
$F_d(X)\le h$, where
\[
 F_d(x)=x-d\arctanh(\tanh x/d),\qquad
 F_d'(x)=\frac{(1-d^{-2})\tanh^2x}{1-d^{-2}\tanh^2x}.
\]
The function is strictly increasing on $(0,\infty)$, and
$F_d(x)\ge x^3/16$ for $0\le x\le1$, because
$\tanh x\ge x/2$ there and $d\ge2$.
For $h\le1/16$, the solution $x_h$ of $F_d(x_h)=h$
satisfies $X\le x_h\le(16h)^{1/3}$.
Since $f_{\beta_c}(x)\le x/d$,
\[
 m_u(h)\le h+\Delta f_{\beta_c}(x_h)
       =x_h+f_{\beta_c}(x_h)
       \le\tfrac32(16h)^{1/3}<4h^{1/3}.
\]
For $h>1/16$ the bound follows from $m_u\le1$.

For the second assertion put $t=(1+\varepsilon)/d<5/8$ and
$F_\beta(x)=x-d\arctanh(t\tanh x)$. Again $F_\beta(X)\le h$.
Comparison in the parameter $t$ gives
\begin{equation}\label{eq:cubic-perturb}
 0\le F_d(x)-F_\beta(x)
       \le\frac{\varepsilon\tanh x}{1-t^2}\le2\varepsilon x.
\end{equation}
Also
\[
 F_\beta'(x)=
 \frac{-\varepsilon+(1+\varepsilon-t^2)\tanh^2x}
      {1-t^2\tanh^2x}>0\qquad(x\ge1),
\]
and $F_\beta(1)\ge1/16-2\varepsilon\ge1/32$.
Thus $h\le1/64$ forces $X<1$, even though
$F_\beta$ need not be increasing near zero.
Using \eqref{eq:cubic-perturb},
$X^3/16\le h+2\varepsilon X$.
If $X\ge8\sqrt\varepsilon$, absorption gives
$X\le(32h)^{1/3}$; otherwise $X<8\sqrt\varepsilon$.
In either case $X\le8\sqrt\varepsilon+4h^{1/3}$.
Since $f_\beta(x)\le tx$ and $\Delta t<2$, \eqref{eq:root} yields
$m_u(h)\le16\sqrt\varepsilon+9h^{1/3}$.
For larger $h$, the stated bound follows from $m_u\le1$.
\end{proof}

\section{Positive-series tilting}

At zero external field put
$M=\sum_{u\in V}\sigma_u$, $Z(h)=\E_0e^{hM}$, and
$\chi_v=\sum_u\E_0[\sigma_v\sigma_u]$.
Global spin-flip symmetry gives
\begin{equation}\label{eq:tilt}
 m_v(h)Z(h)=\E_0[\sigma_v\sinh(hM)]\ge h\chi_v.
\end{equation}
To verify the inequality, expand the hyperbolic sine.
Every term $\E_0[\sigma_vM^{2j+1}]$ is nonnegative by the first
Griffiths inequality after expanding the power and canceling squares.
The series is absolutely convergent since $|M|\le n$.

At criticality Proposition~\ref{prop:response} gives
\[
 \log Z(h)=\int_0^h\sum_um_u(s)\,ds\le3nh^{4/3}.
\]
Combining this with \eqref{eq:tilt}, then choosing
$h=(6n)^{-3/4}$, proves
\begin{equation}\label{eq:critical-chi}
 \max_v\chi_v\le4h^{-2/3}e^{3nh^{4/3}}
                 =4\sqrt{6e}\sqrt n.
\end{equation}
This is a pointwise row bound; an estimate on average magnetization
alone would not supply the same conclusion.

If $\beta\le\beta_c$, \eqref{eq:critical-chi} already gives
\eqref{eq:crossover} with $\varepsilon=0$.
For $0\le\varepsilon=(\Delta-1)\tanh\beta-1\le1/64$,
the second response estimate gives
\begin{align}
 \log Z(h)&\le16n\sqrt\varepsilon\,h+12nh^{4/3},\label{eq:mgf}\\
 \chi_v&\le16(\sqrt\varepsilon/h+h^{-2/3})
                   e^{16n\sqrt\varepsilon h+12nh^{4/3}}.\label{eq:chi-window}
\end{align}
Let $a=n^{1/4}\sqrt\varepsilon$, $b=1+a$, and
$h=n^{-3/4}/(16b)$. The exponent in \eqref{eq:chi-window}
is at most $a/b+b^{-1}=1$.
Its prefactor divided by $16\sqrt n$ is
$16a(1+a)+16^{2/3}(1+a)^{2/3}$.
Using $2a\le1+a^2$ and $(1+a)^2\le2(1+a^2)$ bounds this
by $32(1+a^2)$. Hence
\begin{equation}\label{eq:window-chi}
 \max_v\chi_v\le512e(\sqrt n+n\varepsilon).
\end{equation}
For $\varepsilon>1/64$, the trivial row bound $n$
is already smaller than $512e\,n\varepsilon$.

\section{Fields and pinnings}

We record the exact use of the field comparison, since normalized
influence requires more than a covariance comparison.
Theorem~1.1 of Ding, Song and Sun~\cite{DSS} states that, for a
ferromagnetic Ising model and fields $g$ and $b\ge0$,
\[
 m_v(g+b)-m_v(g-b)\le m_v(b)-m_v(-b).
\]
It allows infinite pinning fields provided
$\min\{|g_w|,b_w\}<\infty$ at each vertex.
For $u\ne v$, take $b_u=+\infty$, $b_w=0$ off $u$,
$g_u=0$, and $g_w=a_w$ off $u$.
The field at a pinned spin is irrelevant. Thus the left side is
$2\Psi_a(u,v)$, while zero-field spin-flip symmetry makes the
right side $2\E_0[\sigma_u\sigma_v]$.
Ferromagnetic monotonicity, including the diagonal convention, gives
\begin{equation}\label{eq:DSS}
 0\le\Psi_a(u,v)\le\E_0[\sigma_u\sigma_v].
\end{equation}
Summing \eqref{eq:DSS} and using
\eqref{eq:critical-chi} or \eqref{eq:window-chi} proves the bounds
with $k=n$.

After pinning vertices $S$ to a configuration $\tau$, the model on
$W=V\setminus S$ has the same couplings on $G[W]$ and fields
$a'_w=a_w+\sum_{s\in S:sw\in E}J_{sw}\tau_s$.
Its maximum degree is at most $\Delta$.
Applying the already field-uniform bounds to this graph of size $k$
proves both theorems.
Finally, $\varepsilon\le A/\sqrt n$ implies
$k\varepsilon\le A\sqrt k$ because $k\le n$.

\section{Glauber dynamics above the threshold}

We apply the entropy localization theorem in
\cite[Theorem~3.21]{CCYZ}.
In its normalization an interaction matrix $J\succeq0$ gives density
proportional to $\exp(\frac12\sigma^\mathsf TJ\sigma+a\cdot\sigma)$.
If $\mathfrak c(s)$ bounds the largest influence eigenvalue at interaction $sJ$
for every field, the approximate tensorization constant is at most
\[
 \exp\left\{\|J\|_2\int_0^1\mathfrak c(s)\,ds\right\}.
\]
Approximate tensorization with constant $K_{\rm AT}$ means that
the entropy of every nonnegative function is bounded by
$K_{\rm AT}$ times the sum of its expected single-site conditional
entropies.

Let $W$ be the weighted adjacency matrix of the given couplings.
The matrix $J=W+\Delta\beta I$ is positive semidefinite and
$\|J\|_2\le2\Delta\beta$; its diagonal shift changes the energy
by a constant. Using the temperature bound $q=s\beta$, the
localization exponent is at most
$2\Delta\int_0^\beta\mathfrak c(q)\,dq$,
now writing $\mathfrak c(q)$ for the bound at temperature upper bound $q$.
Put $d=\Delta-1$ and $b_c=\arctanh(1/d)$.
For $q<b_c$, the subcritical spectral bound from
\cite[Lemma~3.23]{CCYZ} and the critical cap give
\begin{equation}\label{eq:subcritical}
 \mathfrak c(q)\le\min\left\{17\sqrt n,\,
                \frac{\Delta}{d(1-d\tanh q)}\right\}.
\end{equation}
This applies to nonuniform couplings bounded by $q$ as well.
Indeed, \eqref{eq:DSS} bounds their field-dependent influence entries
by their zero-field correlations, GKS bounds those by the uniform
coupling-$q$ correlations, and the Perron eigenvalue is increasing
under entrywise comparison of nonnegative matrices.
Only a spectral bound is used in \eqref{eq:subcritical}.

The exact identity
\[
 \frac{\Delta}{d(1-d\tanh q)}
  =\frac1d+\frac{2}{(\Delta-2)(e^{2(b_c-q)}-1)}
\]
and $e^x-1\ge x$ allow integration up to
$b_c(1-n^{-1/2})$; use $17\sqrt n$ on the remaining interval.
Since $\Delta b_c\le3$,
\begin{equation}\label{eq:below-integral}
 2\Delta\int_0^{b_c}\mathfrak c(q)\,dq
           \le105+\frac{\Delta}{\Delta-2}\log n .
\end{equation}
For $q\in[b_c,\beta]$, Theorem~\ref{thm:crossover} gives
$\mathfrak c(q)\le C[\sqrt n+n(d\tanh q-1)]$.
With $r=d\tanh q-1\le1/64$,
\[
 \frac{dq}{dr}=\frac1{d(1-\tanh^2q)}
       \le\frac{64}{39d},\qquad 2\Delta\frac{dq}{dr}<5.
\]
It follows that
\[
 2\Delta\int_{b_c}^{\beta}\mathfrak c(q)\,dq
       \le5C(\sqrt n\,\varepsilon+n\varepsilon^2/2).
\]
If $\beta<b_c$, simply omit this last integral and retain
\eqref{eq:below-integral}. Thus
\[
 K_{\rm AT}\le e^{105}n^{\Delta/(\Delta-2)}
                         e^{5C(A+A^2/2)}.
\]

In zero field, $\log(1/\mu_{\min})\le n(\log2+\Delta\beta)\le4n$
under $\varepsilon\le1/64$.
Single-site entropy contraction followed by Pinsker's inequality
therefore bounds $T_{\rm mix}$ by a numerical constant times
$nK_{\rm AT}\log(en)$.
This proves the first assertion of Corollary~\ref{cor:mixing}.
For fixed $B$, the last assertion follows from $A\le B\sqrt{\log n}$ and
$A+A^2/2\le1/2+A^2\le1/2+B^2\log n$.
The restriction $\varepsilon\le1/64$ then holds for all sufficiently large $n$.

\end{document}